\documentclass{article}
\usepackage{amsmath,amsfonts,amssymb,amsthm,etoolbox,stmaryrd}

\usepackage{algorithm}
\usepackage[noend]{algpseudocode}
\algrenewcommand\algorithmicthen{\relax}
\algrenewcommand\algorithmicdo{\relax}

\usepackage[colorlinks=true,citecolor=blue,pdfpagemode=UseNone,pdfstartview=FitH]{hyperref}

\allowdisplaybreaks[4]

\renewcommand{\d}{\,\mathrm{d}}
\newcommand{\dd}{\mathrm{d}}

\newcommand{\E}{\mathbb{E}}

\newcommand{\R}{\mathbb{R}}

\newcommand{\EEE}{\mathcal{E}}
\newcommand{\FFF}{\mathcal{F}}
\newcommand{\PPP}{\mathcal{P}}
\newcommand{\SSS}{\mathcal{S}}
\newcommand{\UUU}{\mathcal{U}}
\newcommand{\WWW}{\mathcal{W}}

\theoremstyle{plain}
\newtheorem{theorem}{Theorem}[section]

\newtheorem{lemma}[theorem]{Lemma}

\theoremstyle{definition}

\newtheorem{example}[theorem]{Example}

\theoremstyle{remark}

\title{An optimality property of the Bayes--Kelly algorithm}
\author{Vladimir Vovk}

\begin{document}
\maketitle

\begin{abstract}
  \smallskip
  This note states a simple property of optimality
  of the Bayes--Kelly algorithm for conformal testing
  and poses a related open problem.

 Its version at \url{http://alrw.net} (Working Paper 39)
 is updated most often.
\end{abstract}

\section{Introduction}

This note is motivated by Gr\"unwald et al.'s
RSS discussion paper \cite{Grunwald/etal:arXiv1906}.
The main result of that paper (Theorem~1) is elegant and satisfying,
but in this note I will concentrate on one of its limitations.
The two modest contributions of my comment
is to state a very simple property of optimality of the Bayes--Kelly algorithm
\cite[Sect.~9.2.1]{Vovk/etal:2022book} in conformal testing
and to pose an open problem.
I will try to make the conformal testing part formally self-contained,
but for further intuition and details, see \cite{Vovk/etal:2022book}.

A \emph{test martingale} $S$ in a probability space $(\Omega,\FFF,P)$
equipped with a filtration $(\FFF_n)_{n=0}^{\infty}$
is a nonnegative martingale starting from 1.
In other words, $S=(S_0,S_1,\dots)$,
each $S_n$ is required to be $\FFF_n$-measurable,
$\E_P(S_n\mid\FFF_{n-1})=S_{n-1}$ for all $n\ge1$,
$S_n$ are required to be nonnegative, and $S_0=1$.
We can interpret $\log S_n$ as the amount of evidence
(measured in bits)
found at time $n$ against $P$ as null hypothesis.
One-step counterparts of test martingales are ``e-variables'',
to be introduced shortly.

\section{Gr\"unwald et al.'s result with discussion}

The following is Gr\"unwald et al.'s \cite{Grunwald/etal:arXiv1906}
main result (in its basic form, namely Theorem~1
in \cite[Sect.~2]{Grunwald/etal:arXiv1906}).
The notation used in it will be explained after the statement.

\begin{theorem}[Gr\"unwald et al.]\label{thm:Grunwald}
  Suppose $Q$ is a probability distribution with full support
  and with density $q$,
  and assume
  \begin{equation}\label{eq:condition}
    \inf_{w\in\WWW}
    D(Q\|P_{w})
    <
    \infty.
  \end{equation}
  Then there exists a (potentially sub-) distribution $P$ with density $p$ such that
  \begin{equation}\label{eq:E-star}
    E^*
    :=
    q/p
  \end{equation}
  is an e-variable.
  Moreover, $E^*$ satisfies, essentially uniquely,
  \begin{equation}\label{eq:optimality-1}
    \sup_{E\in\EEE}
    \E_Q(\log E)
    =
    \E_Q(\log E^*)
    =
    \inf_{W\in\WWW}
    D(Q\|P_{W})
    =
    D(Q\|P).
  \end{equation}
  If the $\inf$ is attained,
  so that
  $
    D(Q\|P_{W^*})
    =
    D(Q\|P)
  $,
  then $P = P_{W^*}$.
\end{theorem}

Theorem~\ref{thm:Grunwald} is about a statistical model
$(P_{\theta}\mid\theta\in\Theta)$
such that each $P_{\theta}$ is absolutely continuous
w.r.\ to some underlying measure $\mu$.
This statistical model plays the role of our null hypothesis,
while $Q$ plays the role of the alternative hypothesis.
The notation $D(Q\|P)$ refers to the Kullback--Leibler divergence
between $Q$ and $P$.
The parameter space $\Theta$ is equipped with a $\sigma$-algebra,
$\WWW$ stands for the set of all probability measures on $\Theta$,
and $\EEE$ stands for the set of all \emph{e-variables},
i.e., nonnegative random variables $E$ satisfying
$\E_{P_{\theta}}(E)\le1$ for all $\theta\in\Theta$.
For each $W\in\WWW$, $P_W:=\int P_{\theta} W(\dd\theta)$
is the mixture of $P_{\theta}$ w.r.\ to $W$.
``Essentially uniquely'' is defined in a natural way.

The first equality in~\eqref{eq:optimality-1}
can be interpreted as saying
that $E^*$ is the optimal e-variable under $Q$.

\subsection{Limitation}
\label{subsec:limitation}

An important limitation of Theorem~\ref{thm:Grunwald}, as I see it,
is the condition \eqref{eq:condition}.
This condition is stated in exactly this way
only in Gr\"unwald et al.'s Theorem~1 in its basic form
of \cite[Sect.~2]{Grunwald/etal:arXiv1906},
but analogous conditions are present in all generalizations
given in \cite{Grunwald/etal:arXiv1906}.

To see how restrictive \eqref{eq:condition} is,
consider, following \cite{Ramdas/etal:2022} and \cite[Chap.~9]{Vovk/etal:2022book},
the case of coin tossing.
Formally, the null hypothesis is $(P_{\theta}\mid\theta\in[0,1])$,
where $P_{\theta}=B_{\theta}^{\infty}$
and $B_{\theta}$ is the probability measure on $\{0,1\}$
satisfying $B_{\theta}(\{1\})=\theta$.
The condition \eqref{eq:condition} then means
that $Q$ should be absolutely continuous
w.r.\ to an exchangeable probability measure on $\{0,1\}^{\infty}$.
This condition is violated for interesting $Q$,
such as a Markov $Q$ outside the family $(P_{\theta}\mid\theta\in[0,1])$.
The main alternative hypothesis used in \cite{Ramdas/etal:2022}
is a Jeffreys-type mixture of the Markov measures,
and then \eqref{eq:condition} is also violated.
It is difficult to think of natural cases
where \eqref{eq:condition} would be satisfied.

\subsection{An objection}

A possible objection to the argument of Sect.~\ref{subsec:limitation}
is that infinite sequences are irrelevant in real life,
where we only observe finite sequences.
If instead of the sample space $\{0,1\}^{\infty}$
we consider the sample space $\{0,1\}^n$
for a finite (even if very large) $n$,
the condition \eqref{eq:condition} will be satisfied.
Isn't it all that matters?

The difficulty with this solution
is that the e-variables $E^*_n$
obtained in this way by using Theorem~\ref{thm:Grunwald}
do not have to cohere with each other for different $n$;
for example, they do not have to form a test martingale,
or a test supermartingale, or an e-process.
Optional continuation and stopping become big problems for $E^*_n$,
and $\log E^*_n$ are no longer jointly valid
as the amount of evidence found against the null hypothesis
at time $n$.

To summarise, yes, we can truncate the sequential process
of observing the bits $z_1,z_2,\dots\in\{0,1\}$,
but it matters where we truncate it.
On the other hand,
if we build an e-variable $E^*:\{0,1\}^{\infty}\to[0,\infty]$
for the infinite time horizon,
L\'evy's ``upward'' theorem \cite[Theorem 14.2]{Williams:1991}
will give us a test martingale.
Therefore, the kind of infinity that we need as our time horizon
is the potential infinity, not the actual one.
If we do not know the number of observations in advance
and just would like to have an online measure of evidence
found against the null hypothesis,
Theorem~\ref{thm:Grunwald} does not give us anything useful.
Replacing $E^*:\{0,1\}^{\infty}\to[0,\infty]$
by $E^*:\{0,1\}^N\to[0,\infty]$ for a very large $N$
is also awkward
since the resulting test martingale will typically depend on $N$
(even over the first few steps $n$).

\subsection{Simple solution and its limitation}

It is interesting that the binary case with a given alternative hypothesis
(such as mixed Markov \cite{Ramdas/etal:2022}
or change-point \cite[Sect.~9.2.3]{Vovk/etal:2022book})
admits a simple solution:
just replace~\eqref{eq:E-star} by
\begin{equation*}
  E^*_n
  :=
  \frac
  {Q([Z_1,\dots,Z_n])}
  {\sup_{\theta\in[0,1]}P_{\theta}([Z_1,\dots,Z_n])}
\end{equation*}
for each step $n=0,1,\dots$,
where $[z_1,\dots,z_n]$ stands
for the set of all infinite sequences in $\{0,1\}^{\infty}$
that begin with $z_1,\dots,z_n$,
and $Z_i$ are the random bits whose realizations are the observed bits $z_i$.
This corresponds to replacing the mean $p$ of the densities $p_{\theta}$ of $P_{\theta}$
($p$ is the mean if we assume that the $\inf$ is attained
according to the last statement of Theorem~\ref{thm:Grunwald})
by the supremum of $p_{\theta}$.
Then $E^*_n$ agree with each other
in the sense of forming an e-process
\cite[Theorem 6]{Ramdas/etal:2022}.

However, the idea of replacing $p$ in \eqref{eq:E-star} by $\sup$
does not work outside narrow parametric cases
\cite[Remark~9.8]{Vovk/etal:2022book}.
The following example is still very basic
(in machine learning the task is usually to predict
the labels of complicated objects, such as movies).

\begin{example}\label{ex:failure}
  Fix a finite time horizon $N\gg1$
  and assume that the observations $z_1,\dots,z_N$ are real numbers,
  $z_n\in\R$,
  so that $Z_1,\dots,Z_N$ are random variables.
  The null hypothesis is that of \emph{randomness}:
  $Z_1,\dots,Z_N$ are IID.
  The alternative hypothesis is a continuous probability measure $Q$ on $\R^N$
  (such as a changepoint hypothesis, as in \cite[Remark~9.8]{Vovk/etal:2022book}).
  Then the likelihood ratio of $Q$
  to the maximum likelihood over the null hypothesis is 0;
  indeed, $Q([z_1,\dots,z_N])=0$ and the maximum likelihood is positive,
  namely at least $N^{-N}$
  (it is exactly $N^{-N}$ if the $N$ observations are all different).
  This is worse than useless,
  as the identical 1 is a trivial test martingale.
\end{example}

\section{Conformal testing}

Our book \cite[Part~III]{Vovk/etal:2022book} presents
a general framework, which we call conformal testing,
for testing nonparametric null hypotheses
(first of all the hypothesis of randomness)
in very general situations typical of machine learning.
The procedure does not depend on assumptions such as \eqref{eq:condition}.
The efficiency of conformal testing is demonstrated in empirical studies
reported in \cite[Chap.~8]{Vovk/etal:2022book},
but theoretical results about efficiency have been established
only in toy binary situations \cite[Sect.~9.2]{Vovk/etal:2022book}.
The validity is, however, guaranteed,
in that conformal testing leads to stochastic processes
that are test martingales whenever the observations are IID.

Let me first introduce some terminology and notation.
The \emph{observation space} $\mathbf{Z}$ is a measurable space.
(In the previous section we had $\mathbf{Z}=\{0,1\}$
and then, in Example~\ref{ex:failure}, $\mathbf{Z}=\R$.)
There is an underlying probability space $(\Omega,\FFF,P)$
in the background, but we rarely need it explicitly.
We observe random elements $Z_1,Z_2,\dots$ of $\mathbf{Z}$
(formally, each $Z_n$ is a measurable mapping from $\Omega$ to $\mathbf{Z}$)
with their realized values denoted by $z_1,z_2,\dots\in\mathbf{Z}$.

The null hypothesis considered in this section
(which we call the hypothesis of randomness)
is that the observations $Z_1,Z_2,\dots$ are IID.
(Generalization to other null hypotheses is briefly discussed
in Sect.~\ref{sec:conclusion}.)
The underlying probability space is assumed to be rich enough;
in particular, a probability measure on $\Omega$ making $Z_1,Z_2,\dots$ IID
is assumed to exist
(conformal prediction also needs a sequence of independent and uniformly distributed
$\tau_1,\tau_2,\dots\in[0,1]$ modelling a random number generator).

We let $\mathbf{Z}^{(*)}$ stand for the set of all \emph{bags}
(or \emph{multisets}) $\lbag z_1,\dots,z_n\rbag$
consisting of elements of $\mathbf{Z}$ (with $n=0$ allowed);
the difference between the bag $\lbag z_1,\dots,z_n\rbag$ and the set $\{z_1,\dots,z_n\}$
is that the bag (while still unordered)
can contain several copies of the same element.

A conformal test martingale is determined by two components:
\begin{itemize}
\item
  A \emph{conformity measure} $A$,
  which is a measurable function $A:\mathbf{Z}^{(*)}\times\mathbf{Z}\to\R$.
\item
  A \emph{betting martingale} $B$,
  which is a test martingale in the probability space
  $([0,1]^{\infty},\UUU,U)$ with filtration $(\UUU_n)_{n=1}^{\infty}$,
  where $\UUU$ is the Borel $\sigma$-algebra on $[0,1]^{\infty}$,
  $U$ is the uniform probability measure on $([0,1]^{\infty},\UUU)$,
  and $\UUU_n$ is the $\sigma$-algebra generated by the first $n$ elements
  of the sequences in $[0,1]^{\infty}$.
\end{itemize}
Given these two components,
we define the corresponding conformal test martingale as follows.

The \emph{$n$th conformal p-value} $p_n$ is defined by
\begin{equation}\label{eq:p}
  p_n
  :=
  \frac
  {\left|\{i:\alpha_i<\alpha_n\}\right|+\tau_n\left|\{i:\alpha_i=\alpha_n\}\right|}
  {n},
\end{equation}
where $i=1,\dots,n$,
the \emph{conformity scores} $\alpha_i$ are computed from $z_i$
using the conformity measure $A$ by
\begin{equation}\label{eq:alpha}
  \alpha_i
  :=
  A(\lbag z_1,\dots,z_n\rbag,z_i),
  \quad
  i=1,\dots,n,
\end{equation}
and $\tau_1,\tau_2,\dots$ are independent random variables
that are distributed uniformly on $[0,1]$
(modelling a random number generator).
The \emph{conformal test martingale} (\emph{CTM}) $S$
determined by $A$ and $B$
is the result of applying the betting martingale $B$
to the p-values \eqref{eq:p}:
\[
  S_n:=B_n(p_1,p_2,\dots),
  \quad
  n=0,1,\dots.
\]
The associated $\sigma$-algebras $\FFF_n$ are those generated by $p_1,\dots,p_n$
(in particular, $\FFF_0=\{\emptyset,\Omega\}$).
An equivalent definition of a CTM
is that it is a test martingale in the filtration $(\FFF_n)$
(this follows from, e.g., \cite[Lemma~A3.2]{Williams:1991}).

The property of validity for the conformal p-values
is that they are independent and uniformly distributed on $[0,1]$
under the null hypothesis.
This implies that a CTM is a test martingale under the null hypothesis.

Theoretical results about efficiency are established
in \cite[Chap.~9]{Vovk/etal:2022book}
only in the binary case
and for a specific conformity measure
(the identical one, $A(\lbag z_1,\dots,z_n\rbag,z_i):=z_i$).
In this note we will take a slightly wider approach:
will fix a conformity measure and will then find the optimal betting martingale
for a given alternative hypothesis.
Therefore, we consider a very limited kind of optimality:
\begin{itemize}
\item
  First, we restrict ourselves to the class of conformal test martingales.
\item
  And even within this class, we consider a fixed conformity measure.
\end{itemize}

\subsection{Full alternative hypotheses}

Let $Q$ be a probability measure on $\mathbf{Z}^{\infty}$;
this is our alternative hypothesis about the distribution of $Z_1,Z_2,\dots$.
It is \emph{full} in the sense
of fully determining the distribution of the observations
$Z_1,Z_2,\dots$;
in Sect.~\ref{subsec:shrunk}
we will consider an alternative probability measure
on a poorer $\sigma$-algebra.
Our null hypothesis, as before, is that $Z_1,Z_2,\dots$ are IID.

The following paragraph is a description of the optimal
(in the sense to be described later) under $Q$
CTM with a given conformity measure $A$.
Computationally efficient (or at least more explicit)
versions of this CTM will be referred to as the \emph{Bayes--Kelly algorithm}.

Consider the following Bayesian model
(a statistical model plus a prior distribution on the parameter space).
The parameter space is $\mathbf{Z}^{\infty}$,
and it is equipped with $Q$ as prior distribution.
The element $P_{\zeta}$ of the statistical model
indexed by $\zeta=(z_1,z_2\dots)\in\mathbf{Z}^{\infty}$
is the distribution of the corresponding p-values
defined by \eqref{eq:p} and \eqref{eq:alpha} for a given $\zeta$
(we regard the random number generator $\tau_1,\tau_2,\dots$ used in \eqref{eq:p}
as fixed).
The marginal distribution of the p-values $p_1,p_2,\dots$ is the mixture
\[
  P
  :=
  \int
  P_{\zeta}
  Q(\d\zeta).
\]
The relative increment $S_n/S_{n-1}$ of the betting martingale
on step $n$ is then defined
as the conditional density $f_n$ of $p_n$ given $p_1,\dots,p_{n-1}$
(we will choose a natural version of the conditional density
given by Lemma~\ref{lem:conditional} below)
evaluated at the realized $p_n$.
Knowing $S_n/S_{n-1}$ defines the betting martingale,
since we know that its starting value is $S_0=1$.
The \emph{Bayes--Kelly CTM} is determined by $A$ and this betting martingale.

\begin{lemma}\label{lem:conditional}
  A conditional density of $p_n$ given $p_1,\dots,p_{n-1}$ exists.
  There is a version $f_n$ of the conditional density
  that is constant over each of the intervals
  \begin{equation}\label{eq:intervals}
    \begin{aligned}
      &[i/n,(i+1)/n), \quad i=0,\dots,n-2,\\
      &[(n-1)/n,1].
    \end{aligned}
  \end{equation}
\end{lemma}

\begin{proof}
  For a fixed $\zeta\in\mathbf{Z}^{\infty}$,
  $P_{\zeta}$ generates independent p-values $p_1,p_2,\dots$,
  and each $p_n$ (defined by \eqref{eq:p})
  is distributed uniformly on an interval $[n_*/n,n^*/n]$
  for some $n_*\in\{0,\dots,n-1\}$ and $n^*\in\{i+1,\dots,n\}$.
  Namely,
  \begin{equation}\label{eq:counts}
    \begin{aligned}
      n_* &= \left|\{i\in\{1,\dots,n\}:\alpha_i<\alpha_n\}\right|\\
      n^* &= \left|\{i\in\{1,\dots,n\}:\alpha_i\le\alpha_n\}\right|
    \end{aligned}
  \end{equation}
  in the notation of \eqref{eq:p}.
  Let $f_n^{\zeta}$ be the density for $p_n$ under $Q$
  conditional on knowing $\zeta$
  and w.r.\ to the uniform probability measure on $[0,1]$.
  This is a piecewise constant function,
  which we assume taking constant values
  on the intervals \eqref{eq:intervals}.
  On each of these intervals,
  $f_n^{\zeta}$ takes values in $[0,n]$.
  Then $f_n$, as the integral of $f_n^{\zeta}$
  w.r.\ to the posterior distribution on $\zeta$
  (with prior $Q$ and after observing $p_1,\dots,p_{n-1}$),
  exists and integrates to 1,
  the latter following from Fubini's theorem.
\end{proof}

In \cite[Chap.~9]{Vovk/etal:2022book},
we spell out the details of the Bayes--Kelly algorithm
in two special (binary) cases:
\begin{itemize}
\item
  changepoint alternatives,
\item
  Markov alternatives, following \cite{Ramdas/etal:2022}.
\end{itemize}
In both cases, the procedure is computationally efficient
(while the computational efficiency of Algorithm~\ref{alg:BK}
described below
is unclear).
We also prove its general optimality properties
(i.e., without the \emph{a priori} restriction to conformal testing).

\begin{algorithm}[bt]
  \caption{Bayes--Kelly algorithm (continuous version)}
  \label{alg:BK}
  \begin{algorithmic}[1]
    \State $S_0:=1$
    \State $\Sigma:=\mathbf{Z}$\label{l:Sigma_0}
    \For{$n=1,2\dots$:}
      \State Set $f_n$ to the density of the pushforward of $Q_{\Sigma}$
          under \eqref{eq:p}--\eqref{eq:alpha}\label{l:pushforward}
      \State Read $z_n\in\mathbf{Z}$
      \State Compute $\alpha_1,\dots,\alpha_n$ as per~\eqref{eq:alpha}
      \State Read $\tau_n\in[0,1]$
      \State Compute $p_n$ as per~\eqref{eq:p}
      \State $S_n:=S_{n-1}f_n(p_n)$
      \For{$(z_1,\dots,z_n)\in\Sigma$:}\label{l:second-for}
        \State Compute $\alpha_1,\dots,\alpha_n$ as per~\eqref{eq:alpha}
        \State Compute $n_*$ and $n^*$ as per~\eqref{eq:counts}\label{l:assume}
        \If{$p_n\notin[n_*/n,n^*/n]$:}
          \State Remove $(z_1,\dots,z_n)$ from $\Sigma$
        \EndIf
      \EndFor
      \State Update $\Sigma:=\Sigma\times\mathbf{Z}$
    \EndFor
  \end{algorithmic}
\end{algorithm}

Algorithm~\ref{alg:BK} is a version of the Bayes--Kelly algorithm
that works for a conformity measure $A$
that is continuous under the true data-generating distribution;
in fact, it is sufficient to assume that,
for all $n$ and almost all $z_1,\dots,z_n$,
the $n$ conformity scores \eqref{eq:alpha} are all different.

Algorithm~\ref{alg:BK} maintains a set $\Sigma$ of sequences $(z_1,\dots,z_n)$
compatible with the p-values $p_1,\dots,p_{n-1}$ observed so far;
it is initialized to $\Sigma:=\mathbf{Z}$ for $n=1$ (line \ref{l:Sigma_0}).
The notation $Q_{\Sigma}$ in line~\ref{l:pushforward}
stands for the conditional distribution of the first $n$ observations
generated from $Q$ given that they belong to $\Sigma$.
For $n=1$ we have $f_1:=1$.
For this and other $n$, $f_n$ is the pushforward of $Q_{\Sigma}$
under the mapping of the type $\mathbf{Z}^n\to[0,1]$
defined in two steps:
first we apply \eqref{eq:alpha} to the input $(z_1,\dots,z_n)$
(ranging freely over $\mathbf{Z}^n$)
and then we apply \eqref{eq:p}.
Similarly,
the variables $z_1,\dots,z_n$ in the second \textbf{for} loop
(starting in line~\ref{l:second-for})
in Algorithm~\ref{alg:BK} are local ones;
they range freely over $\mathbf{Z}$
and do not interfere with the global $z_1,\dots,z_n$,
which are the first $n$ observations.
Finally, the $\alpha_1,\dots,\alpha_n$ inside that loop
are local variables that are completely separate
from the global variables with the same name.
In line~\ref{l:assume} we may assume $n^*=n_*+1$
(this can be violated only with probability zero).

It should be clear how to drop the assumption of continuity
of the conformity measure $A$ in Algorithm~\ref{alg:BK}:
a sequence $z_1,z_2,\dots$ leading to
\[
  k_n
  :=
  \left|
    \left\{
      i=1,\dots,n:
      \alpha_i=\alpha_n
    \right\}
  \right|
  >
  1
\]
should have its weight multiplied by $1/k_n$ at step $n$
after observing $p_n$ that is compatible with this sequence.

The optimality property of the Bayes--Kelly algorithm
is akin to Gr\"unwald et al.'s one,
namely to the first equality in \eqref{eq:optimality-1}.
The next theorem will spell it out.
In its statement,
$\SSS_A$ is the class of all CTMs based on a conformity measure $A$,
and $\PPP$ (representing the null hypothesis of randomness)
consists of all $P^{\infty}$,
$P$ ranging over the probability measures on $\mathbf{Z}$.

\begin{theorem}\label{thm:optimality}
  Fix a conformity measure $A$
  and an alternative hypothesis $Q$.
  At each step $N$,
  the Bayes--Kelly CTM $S^*$ attains the maximum of $\E(\log S_N)$
  among all CTMs $S$ based on $A$:
  \begin{equation}\label{eq:optimality-2}
    \sup_{S\in\SSS_A}
    \E_Q(\log S_N)
    =
    \E_Q(\log S^*_N)
    =
    D(A_*Q\|A_*P),
  \end{equation}
  where $A_*Q$ is the pushforward of $Q$
  under the mapping \eqref{eq:p}--\eqref{eq:alpha}
  of generating the p-values $p_1,\dots,p_N$,
  $P$ is an arbitrary element of $\PPP$,
  and $A_*P$ (which is the uniform distribution on $[0,1]^N$)
  is defined analogously to $A_*Q$.
\end{theorem}

Notice the similarity between \eqref{eq:optimality-1} and \eqref{eq:optimality-2}
(and we can also add $\inf_{P\in\PPP}$
in front of $D(A_*Q\|A_*P)$ in \eqref{eq:optimality-2}).
The property of optimality given in Theorem~\ref{thm:optimality}
is typical of Bayesian algorithms.

\begin{proof}[Proof of Theorem~\ref{thm:optimality}]
  The optimization problem $\E_Q(\log S_N)\to\max$
  decomposes into the sequence of problems $\E_Q(\log(S_n/S_{n-1}))\to\max$
  for $n=1,\dots,N$ and for given $p_1,\dots,p_{n-1}$.
  It suffices to apply \cite[Lemma~9.6]{Vovk/etal:2022book}.
\end{proof}

Unlike Theorem~\ref{thm:Grunwald},
Theorem~\ref{thm:optimality} gives e-variables $S^*_n$
that are coherent in the sense of forming a test martingale.
The reason for this is that using conformal p-values
reduces the massive null hypothesis of randomness
to the simple hypothesis of uniformity.

\subsection{Shrunk alternative hypotheses}
\label{subsec:shrunk}

This subsection will be very speculative.

The Bayes--Kelley algorithm does not cover some interesting cases.
The conformity measures $A$ can be very complex,
and then the problem of designing an optimal betting martingale
becomes infeasible.
It is even possible for $A$ to have an element of intelligence in it
\cite[Sect.~8.6.3]{Vovk/etal:2022book}.
For example, $A$ can be based on a deep neural network
as an underlying algorithm \cite[Sect.~4.3.3]{Vovk/etal:2022book}.
In this case the Bayes--Kelly algorithm will be infeasible
even if we fix an alternative probability measure $Q$ on $\mathbf{Z}^{\infty}$.
(And we are unlikely to have such a probability measure $Q$ in the first place
if we are contemplating the use of such a conformity measure~$A$.)

In the case of such a ``quasi-intelligent'' conformity measure,
to find a suitable betting martingale $B$,
we might consider an alternative distribution on the p-values $p_1,p_2,\dots$
(whose distribution is uniform under the null hypothesis)
rather than an alternative distribution on the observations
or (which is not very different) on the underlying probability space.
Let us call a probability measure on $[0,1]^{\infty}$
interpreted as alternative distribution on the p-values
a \emph{shrunk alternative}
(in the spirit of ``filtration shrinkage'',
a popular topic of research in probability theory).
The betting martingale $B$ that is optimal in a natural sense
will then be the likelihood ratio of the shrunk alternative
to the null hypothesis of the uniform distribution on $[0,1]^{\infty}$.

How do we choose the shrunk alternative $Q$?
A natural approach is simply to try and make it
as large as possible in an attempt to approximate
the universal distribution
in the sense of the algorithmic theory of randomness
(see, e.g., \cite[Sect.~5]{Vovk/Vyugin:1994},
which calls it ``\emph{a priori} semidistribution'',
or \cite[Appendix A]{Vovk:arXiv2305},
which considers a non-sequential setting).
Such a betting martingale may be called ``quasi-universal'';
see \cite[Appendix B]{Vovk:arXiv2305} for a further discussion
in a non-sequential setting.
The quasi-universal betting martingale
is designed to approximate the universal supermartingale
(see, e.g., \cite[Sect.~3]{Vovk/Vyugin:1994}).

To summarise, an appealing informal choice is:
\begin{itemize}
\item
  a quasi-intelligent conformity measure;
\item
  a quasi-universal betting martingale.
\end{itemize}
To make testing methods based on these choices computationally efficient,
we might need modifications, e.g.,
in the direction of inductive conformal prediction
\cite[Sect.~4.2]{Vovk/etal:2022book}.

When designing a quasi-intelligent conformity measure,
we still need an objective function, perhaps informal.
In conformal prediction \cite[Part~I]{Vovk/etal:2022book}
a typical informal objective function
is the conformity measure's sensitivity to unusual observations
(and so detecting their ``nonconformity''),
which leads to smaller p-values for non-IID data.
Motivated by this informal objective function,
in the first edition of \cite{Vovk/etal:2022book}
(Sect.~7.1 of the first edition)
we only considered betting martingales $S$
such that, for each $n$, $S_n$ is a decreasing function
of the $n$th p-value $p_n$.
However, later \cite[Sect.~8.6.1]{Vovk/etal:2022book}
it turned out that allowing non-decreasing $S_n$
greatly improves the performance of conformal test martingales
on benchmark datasets.
Now it appears that our informal objective function
in designing a good conformity measure
should be not to minimize the p-value
(as it usually is in conformal prediction \cite[Sect.~3.1]{Vovk/etal:2022book})
but to come up with the most informative conformity scores $\alpha_i$
capturing the most relevant features of~$z_i$.

\section{Conclusion}
\label{sec:conclusion}

In this note we concentrated on the case of the hypothesis of randomness
as the null.
In fact conformal testing is applicable
to many other ``online compression models'';
see \cite[Part~IV]{Vovk/etal:2022book}.
Examples include partial exchangeability, Gaussian,
hypergraphical (useful for causal inference), and Markov models.

The most obvious open problem arising in connection with Theorem~\ref{thm:optimality}
is whether there is a way of extending the Bayes--Kelly algorithm
and \eqref{eq:optimality-2} to the case when the conformity measure $A$
also needs to be chosen optimally.
When only given the alternative hypothesis $Q$,
how do we choose the pair $(A,B)$,
where $B$ is the betting martingale,
optimally?
In this note we have only discussed how to choose $B$ optimally
given $Q$ and $A$.

Of course, it is not surprising that Theorem~\ref{thm:Grunwald}
of \cite{Grunwald/etal:arXiv1906} has limitations
(and the authors of \cite{Grunwald/etal:arXiv1906} discuss some),
but it is a great first step,
and it opens up interesting directions of further research.

\subsection*{Acknowledgements}

I am grateful to attendees of the RSS meeting
where \cite{Grunwald/etal:arXiv1906} was presented,
in particular Peter Gr\"unwald, Wouter Koolen, and Johannes Ruf,
for useful discussions.

\end{document}